\documentclass[journal,comsoc]{IEEEtran}

\usepackage[T1]{fontenc}

\usepackage[nocompress]{cite}

\usepackage{amsthm, amsmath,amssymb,amsfonts}
\usepackage{algorithmic}
\usepackage{graphicx}
\usepackage{textcomp}
\usepackage{xcolor}
\usepackage{makecell}
\usepackage{caption}
\usepackage{comment}
\usepackage{booktabs}
\usepackage{algorithm}
\usepackage{float}
\usepackage{subfigure}

\newtheorem{property}{Property}
\newtheorem{ppst}{Proposition}
\newtheorem{theorem}{Theorem}

\ifodd 0
\newcommand{\rev}[1]{{\color{blue}#1}} 
\newcommand{\com}[1]{\textbf{\color{red} (COMMENT: #1)}} 
\newcommand{\comg}[1]{\textbf{\color{green} (COMMENT: #1)}}
\newcommand{\response}[1]{\textbf{\color{magenta} (RESPONSE: #1)}} 
\else
\newcommand{\rev}[1]{#1}
\newcommand{\com}[1]{}
\newcommand{\comg}[1]{}
\newcommand{\response}[1]{}
\fi

\begin{document}
\title{Joint Multi-User DNN Partitioning and Computational Resource Allocation for Collaborative Edge Intelligence}

 \author{Xin Tang,
         Xu Chen,
         Liekang Zeng,
         Shuai Yu,
         and~Lin Chen
         \\School of Data and Computer Science, Sun Yat-sen University, Guangzhou, China
 }

\maketitle
\begin{abstract}
  Mobile Edge Computing (MEC) has emerged as a promising supporting architecture providing a variety of resources to the network edge, thus acting as an enabler for edge intelligence services empowering massive mobile and Internet of Things (IoT) devices with AI capability.
  With the assistance of edge servers, user equipments (UEs) are able to run deep neural network (DNN) based AI applications, which are generally resource-hungry and compute-intensive, such that an individual UE can hardly afford by itself in real time.
  However the resources in each individual edge server are typically limited. Therefore, any resource optimization involving edge servers is by nature a resource-constrained optimization problem and needs to be tackled in such realistic context.
  Motivated by this observation, we investigate the optimization problem of DNN partitioning (an emerging DNN offloading scheme) in a realistic multi-user resource-constrained condition that rarely considered in previous works.
  Despite the extremely large solution space, we reveal several properties of this specific optimization problem of joint multi-UE DNN partitioning and computational resource allocation.
  We propose an algorithm called Iterative Alternating Optimization (IAO) that can achieve the optimal solution in polynomial time.
  In addition, we present rigorous theoretic analysis of our algorithm in terms of time complexity and performance under realistic estimation error.
  Moreover, we build a prototype that implements our framework and conduct extensive experiments using realistic DNN models, whose results demonstrate its effectiveness and efficiency.
\end{abstract}

\section{Introduction}\label{sec:introduction}
\IEEEPARstart{D}{riven} by the breakthroughs in deep learning, recent years have witnessed a booming of artificial intelligence (AI) applications and services, ranging from face recognition~\cite{deepfacereco}, video analytics~\cite{surveillancesurvey} to natural language processing~\cite{nlpsurvey}. In the meantime, with the proliferation of mobile Internet and Internet of Things (IoT), a large number of mobile and IoT devices are deployed at the network edge and generate a huge amount of data~\cite{edgesurvey,iotdata}.
To fully unleash the potential of these mobile and IoT big data, there is an urgent need to push the AI capability to the network edge for real-time data processing. In this context, edge AI or edge intelligence is emerging as a promising paradigm to fulfill such demand~\cite{edgeai}.

Due to its great potential, edge intelligence is starting to attract extensive research attention~\cite{coin, ei1, ei2}. By providing a variety of resources to user equipments (UEs) in close proximity~\cite{edgesurvey,ec}, edge servers are able to assist UEs in running resource-hungry applications. However, running one of the most fundamental AI tools in real time at the network edge, deep neural network (DNN), is still a challenging task.

\rev{There are mainly two kinds of approaches to deploy a DNN model on mobile and IoT devices:}

\rev{ 1) Leveraging model compression mechanism to reduce the model trading accuracy for less computation, such as model pruning~\cite{nest,DBLP:conf/nips/GuoYC16,DBLP:journals/corr/HanMD15,DBLP:conf/nips/HanPTD15} and weights/activations quantization~\cite{DBLP:conf/nips/CourbariauxBD15,DBLP:journals/corr/CourbariauxB16,DBLP:conf/icml/GuptaAGN15,DBLP:conf/nips/HubaraCSEB16}.
Generally, the accuracy of a compressed model is not geranteed theoretically. Besides, model compression methods are model-specific that highly depending on the architecture of DNN~\cite{mednn}. }

\rev{ 2) Deploying DNN models in a distributed manner by offloading part or entire computation tasks onto other devices to reduce the local workload~\cite{neurosurgeon,jdnn,hermes,DBLP:conf/ccece/ShahzadS16,energyoptimal}. It's worth noting that model compression tricks can be jointly applied together with offloading mechanisms (e.g.~\cite{mednn,coin}).}

As a matter of fact, traditional computation offloading schemes of edge computing can hardly deal with DNN-related tasks in real time due to its large data volume and high computation overhead. To address the above challenge, DNN partitioning~\cite{neurosurgeon,jdnn,coin} is proposed, which, instead of offloading the entire task onto an edge server, handles a neural network model as a sequential connection of several independent layers and executes them separately on multiple devices.
By carefully configuring layers, DNN partitioning can limit the communication between the UEs and their serving servers, thus reducing the corresponding latency.

Frameworks like Edgent\cite{coin} and Neurosurgeon\cite{neurosurgeon} choose to use a binary DNN partitioning scheme. Eshratifar \textit{et al.} goes one step further such that the execution of each layer can be independently decided to be placed at either the local device or cloud~\cite{jdnn}. The propositions in \cite{hermes,dyoff,energyoptimal} apply general offloading mechanisms for tasks with specific structural dependency like DNN Models. However, all of these works focus on the single-user case.

When it comes to the more general multi-user case, computation offloading needs to be handled jointly with resource allocation among users. In this regard, there exists a handful of propositions (cf.~\cite{offloadinggame, cooperativemanage, yang, zhao, jointo}) on the multi-user computation offloading and resource allocation for general tasks. However, most of them only consider offloading the \textit{entire task} to edge or cloud servers, or optimizing an offloading ratio of several tasks.
As for the multi-user multi-level offloading, Du~\textit{et al.} discusses joint resource management for multiple devices, but only a fixed number of offloading decisions can be supported such that the whole task can be either offloaded to an edge or cloud server~\cite{minmax}.
\cite{nestdnn} adopts the idea of DNN filter pruning for dynamic model resizing in order to reduce the total computation offloading cost. However, since filter pruning would damage the inference accuracy, a good balance between cost reduction and accuracy loss is difficult to derive.

\emph{To the best of our knowledge, none of the previous work has considered a joint multi-user DNN partitioning based multi-level offloading and the related computational resource allocation problem, which is the focus of our work.}

Motivated by the above observation, we embark in this paper on the study of joint multi-UE DNN partitioning and the related edge computational resource allocation.
Our long-term vision is to build an intelligent edge that can satisfy vast demands on DNN-based edge AI applications from heterogeneous UEs. More specifically, multiple UEs cooperate with a resource-constrained edge server, where each UE can make a DNN partitioning decision on its own DNN model and the edge server efficiently allocates its computational resources to different UEs to accelerate the execution of their offloaded DNN layers at the edge server. \rev{For many edge intelligence applications such as smart manufacturing, intelligent transportation and unmanned aerial vehicles,  boosting the real-time performance is a key primary metric to pursuit particularly when there are multiple devices that run concurrent DNN tasks and compete for the limited computing resources on the edge server. Thus, in this study we aim to minimize the maximum DNN execution latency among all the devices, in order to reduce the global latency and enhance the system-wide real-time performance.} We develop a resource allocation framework minimizing the maximum latency of the DNN executions among the UEs.

The main contributions of our work are articulated as follows.

\textbf{Problem formulation and framework}. We investigate the multi-user edge intelligence application scenario and develop a framework of joint DNN partitioning and computational resource allocation. We advocate using data-driven correction function for more accurate multi-core computing capability modeling, based on which  we formulate a max-min optimization problem.

\textbf{Algorithm design and analysis}. We analyze the structural properties of the multi-user joint DNN partitioning and computational resource allocation problem, based on which we further develop an algorithm called Iterative Alternating Optimization (IAO) that can achieve the optimal solution in polynomial time. We present rigorous theoretic analysis of our algorithm in terms of time complexity and performance bound under realistic estimation error.

\textbf{Prototyping and experiment}. We build a prototype that implements our algorithm. We implement several popular DNN models with heterogeneous UEs including Raspberry Pis and NVIDIA Jetson Nanos. We conduct extensive experiments whose results demonstrate the effectiveness and efficiency of our approaches over existing benchmarks.

The rest of this paper is organized as follows.
In Section~\ref{section:systemmodel}, we introduce the basic concepts of DNN partitioning and formulate this problem as an optimization problem.
In Section~\ref{section:algorithmdesign}, we propose an iterative alternating optimization algorithm to efficiently solve this problem.
Besides, detailed analysis and proofs are given in this section.
In Section~\ref{section:evaluation}, we evaluate the performance of our framework on a realistic prototype with extensive experiments.
Additionally, we analyze the efficiency of our algorithms against the increase of the problem size.
Finally, we draw conclusions in Section~\ref{section:conclusion}.

\begin{figure}[t]
    \centerline{\includegraphics[width=9.5cm]{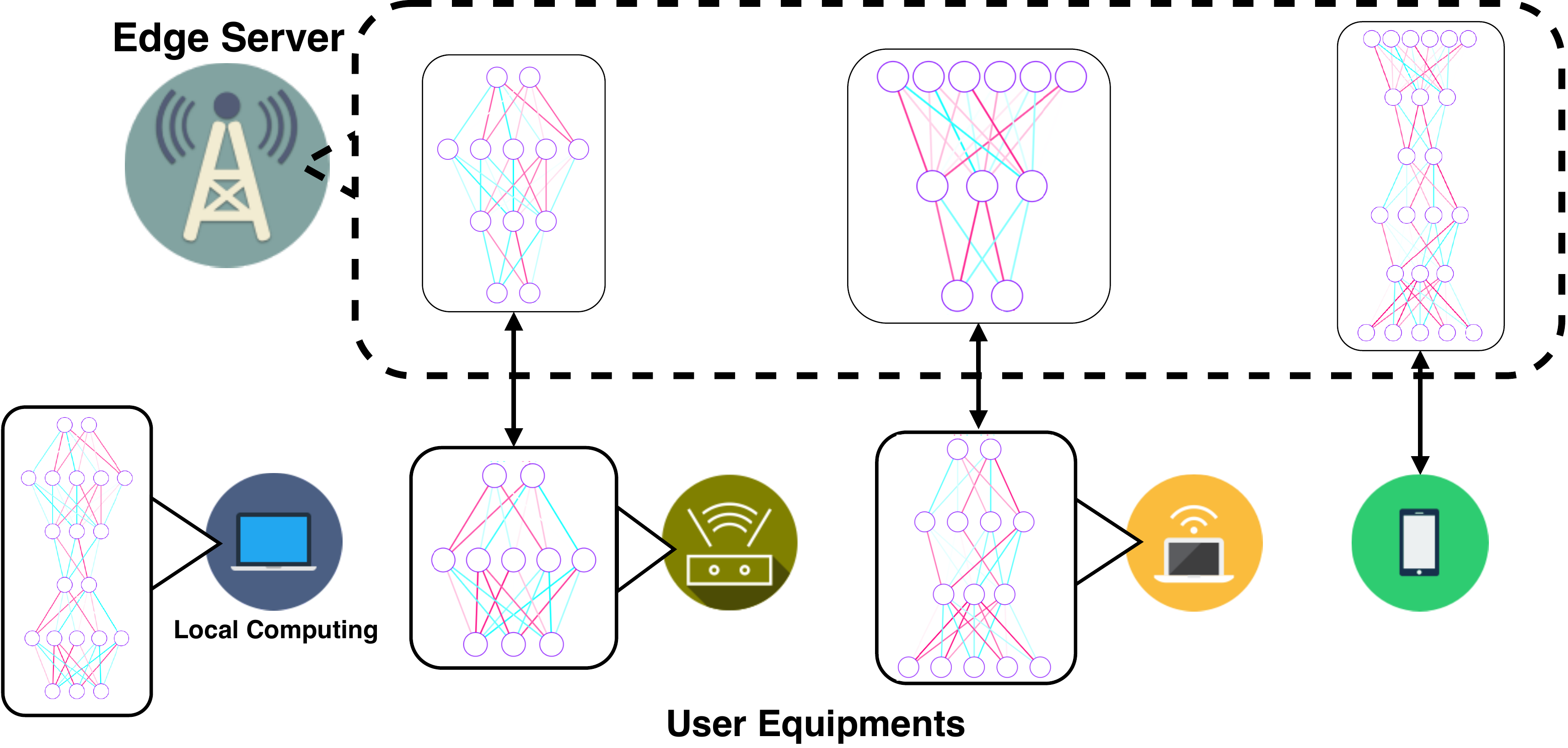}}
    \caption{Illustration of DNN Partitioning: different UEs can choose different partition patterns for DNN task offloading.}
    \label{fig:dnnpartitioning}
    \end{figure}


\section{System Model and Problem Formulation}\label{section:systemmodel}

We consider that a set $\mathcal{N}$ of UEs executing their DNN inference tasks with the assistance of a resource-constrained edge server in proximity.
As shown in Fig.~\ref{fig:dnnpartitioning}, each UE is equipped with a pre-trained DNN model.
To accelerate the execution process of a DNN inference task, each UE is able to make a layer-level offloading decision using a mechanism called \textbf{DNN Partitioning} (will be discussed in the next subsection).
A replicate of virtual server, which can be a container or virtual machine hosted in the edge server, equipped with a fixed amount of computational resources, can assist its associated UEs by allocating computational resources of the edge server and executed concurrently the offloaded tasks. Tab.~\ref{tab:notation} lists the main notations used in this paper.

\rev{Note that as an initial thrust towards efficient algorithm design for multi-user edge intelligence, in this paper we focus on the computational resource allocation of the edge server, which would be the key bottleneck resource for supporting computationally-intensive DNN tasks by multiple users. The joint allocation of multi-dimensional resources consisting of computing, communication and energy would be much more challenging due to the combinatorial nature of the resource allocation decision spaces. Nevertheless, in practice to avoid heavy overhead due to frequent system configurations,  the allocation of other resources such as bandwidth and energy is typically carried out at a slower time-scale than that of computational resource allocation, and hence the proposed algorithms in this paper can also be useful for an integrated system resource allocation framework design in the future work.}

\begin{table}
\caption{Main notations}
\vspace{-0.8em}
\begin{tabular}{p{0.5cm} p{7.5cm} }
\toprule
$\mathcal{N}$ & UE set, $n=|{\mathcal{N}}|$ \\
$s_i$& DNN partitioning and offloading decision variable of UE $i$\\
$k_i$& number of logical layers of DNN inference task at UE $i$\\
$C_i^D$& computational capability of UE $i$\\
$C_{min}$& Minimum Computational Resource Unit (MCRU)\\
$\beta$& total number of MCRU on the edge server\\
$f_i$& number of assigned MCRU for each UE $i$\\
$B_i^{ul}$&upload bandwidth of UE $i$ \\
$B_i^{dl}$& download bandwidth of UE $i$ \\
$X_{i,s_i}$ & amount of computation before separation point $s_i$ at UE $i$\\
$Y_{i,s_i}$ & amount of computation after separation point $s_i$ at UE $i$\\
$M_{i,s_i}$ & intermediate output after separation point $s_i$ at UE $i$\\
$p$ & decremental factor for IAO-DS\\
\bottomrule
\end{tabular}
\label{tab:notation}
\end{table}

\subsection{DNN Partitioning}\label{dnnpartitioning}
Generally, most modern DNN models are constructed by several basic elements called ``layers''. The identity of a ``layer'' represents a combination of computations with similar effort toward a set of input (e.g., the convolution layer in CNN).

In this paper, we consider layers as the atomic elements of a DNN model. To deal with different DNN architectures, we first abstract them as sequential computational graphs by considering parallel layers or layers with shortcut connections as a composite entity called logical layer (e.g., a residual block, as shown in Fig.~\ref{logical}).
Thus, a DNN model is treated as a sequential connection of several logical layers.
To partition the computation of a DNN model, we divide the logical layers into two categories, those executed locally, and those executed on the edge server.
Only intermediate outputs are transferred from a UE to its edge server\footnote{Model compressions and parameter pruning are also widely used to reduce the size of edge AI applications, and DNN partitioning can be applied based on the compressed or pruned edge AI models as well.}.

Specifically, the DNN model of each UE $i$ is composed of $k_i$ logical layers.
An offloading decision at UE $i$ can thus be modeled with an integer variable $s_i \in \{0,1,2,...,k_i\}$, indicating that the layers $0$ to $s_i$ are executed locally while the rest of layers are offloaded to the edge server. The two degenerated cases $s_i = k_i$ and $s_i=0$ correspond to the cases where the entire computation task is run locally at the UE and offloaded to the edge server, respectively.

\begin{figure}[t]
    \centerline{\includegraphics[width=1.8in]{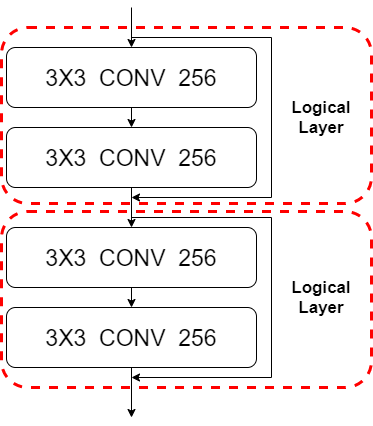}}
    \caption{Abstraction of residual block consisting of multiple layers as a logical layer.}
    \label{logical}
\end{figure}

\subsection{DNN Execution Latency Model}\label{lp}
Once the DNN partitioning decision is made, the execution time for a UE is composed of four parts, which are described as follows.

\subsubsection{Execution Time on UE}
Let $X_{i,s_i}$ denote the amount of the local computation of DNN model before separation point $s_i$ of UE $i$; let $C_i^D$ denote the computational capability of UE $i$.
The local execution time is given by:
$$t_{local}=\frac{X_{i,s_i}}{C_i^D}.$$

\subsubsection{Intermediate Output Transmission Time}
Let $M_{i,s_i}$ denote the data size of intermediate output to be transferred after DNN partitioning; let $B_i^{ul}$ denote the upload bandwidth of UE $i$.
The transmission time of intermediate output to the offloaded edge server is given by:
$$t_{upload}=\theta(k_i-s_i)\frac{M_{i,s_i}}{B_i^{ul}},$$
where $\theta(x)$ is the unit step function defined as follows:
\[
\theta(x)=\left\{
\begin{array}{rcl}
1 & & {x > 0,}\\
0 & & {x \leq 0.}\\
\end{array} \right .
\]
When $s_i\le k_i$, the inference is executed locally without any communication between UE and the edge server, i.e., $t_{upload}=0$.

\subsubsection{Execution Time on Edge Server}

It is generally assumed in the literature that the task execution speeds scales linearly w.r.t. the amount of allocated computation resource. Hence, the execution time is inversely proportional to the amount of allocated computational resource.
\begin{figure}[t]
    \centerline{\includegraphics[width=7cm]{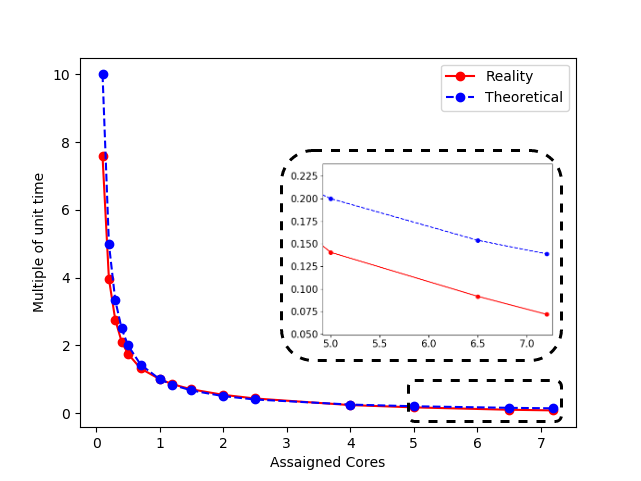}}
    \caption{The theoretical and real execution time of multiple cores using a VGG19 inference task and Intel$^\circledR$ Xeon$^\circledR$ W-2145 CPU. The result illustrates the non-linearity in multicore concurrent efficiency.}
    \label{fig:linearity}
    \end{figure}

However, this generic assumption is no longer valid for the DNN inference tasks, particularly with multi-core CPUs.
To verify this, we carry out experiments by executing a standard DNN inference task on docker containers\cite{docker}, which are widely applied in server virtualization,  with different numbers of CPU cores.
As shown in Fig. \ref{fig:linearity}, there is up to $44\%$ error in execution time between theory and experiment (when there are 7.2 cores).

In our model, we apply a data-driven approach to mitigate the above gap in the case of multi-core CPUs. Our approach is based on the practical model with discrete computational resources. Let $C_{min}$ denote the minimum computational resource unit (MCRU). Let $\beta$ denote the total number of MCRU. Let $f_i$ denote the number of MCRU allocated to UE $i$. Note that $\sum_{i\in {\mathcal{N}}} f_i \le \beta$.

To model the non-linearity in the execution time, we express the execution time on edge as
$$t_{edge}=\theta(k_i-s_i)\frac{Y_{i,s_i}}{\gamma(f_i)C_{min}},$$
where $Y_{i,s_i}$ denotes the amount of computation offloaded by UE $i$ to the edge after separation point $s_i$, $\gamma(f)$ is a compensation function we introduce to fit the execution time in the multi-core CPU case\footnote{For the ease of implementation, we focus on the multi-core CPU case in this study. Nevertheless, similar method can be also applied in the GPU case.}. Note that in the single-core case, $\gamma(f)$ is degenerated to $f$. In practice, we can conduct data-driven fitting to obtain $\gamma(f)$. A common practice is to gather a large samples of runtime profiles of the edge server and then carry out nonlinear regression-based (e.g., regression trees \cite{hutter2014algorithm}) estimation for $\gamma(f)$. In the following algorithm design and analysis, we only assume that the compensation function $\gamma(f)$ (equivalently, the  effective computing capability $\gamma(f)C_{min}$) increases with the computational resource $f$.
\begin{figure}[t]
    \centerline{\includegraphics[width=7cm]{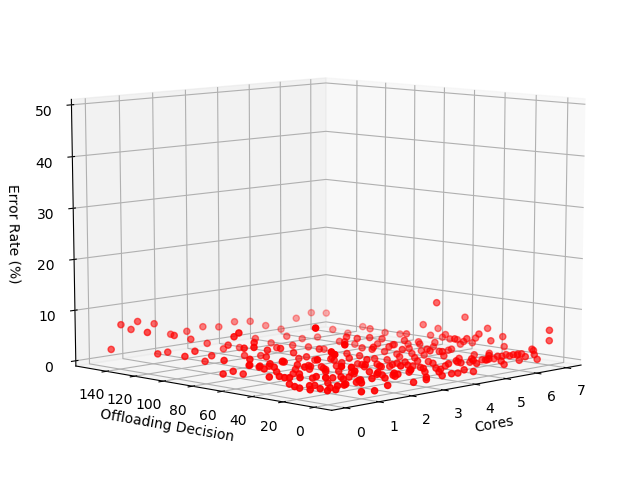}}
    \caption{Performance of Latency Model investigated by a MobilenetV2 instance}
    \label{fig:predict}
    \end{figure}

\subsubsection{Final Result Transmission Time}
Similar to the transmission time of intermediate output, we can derive the transmission time of the final results as
$$t_{download}=\theta(k_i-s_i)\frac{M_{i,{k_i}}}{B_i^{dl}},$$
where $M_{i,{k_i}}$ denotes the size of the final results of offloaded DNN, $B_i^{dl}$ denotes the down-link bandwidth.

By combining the above analysis, we can derive the overall latency of each UE $i$ as follows:

\begin{multline}
T_i(s_i,f_i)=\underbrace{\frac{X_{i,s_i}}{C_i^D}}_{Local} + \\
\theta(k_i-s_i)\left[\underbrace{\frac{M_{i,s_i}}{B_i^{ul}}}_{Upload} +\underbrace{\frac{Y_{i,s_i}}{\gamma(f_i)C_{min}}}_{on\ Edge} +\underbrace{\frac{M_{i,{k_i}}}{B_i^{dl}}}_{Download}\right].
\label{eq:M1}
\end{multline}

We conclude this subsection by checking our model on the DNN execution latency in real-world environment. Our experiment results, shown in Fig. 4, demonstrate that: 1) the average relative error between the latency estimated by our model and the actual latency (i.e. $|T^{estimated}-T^{actual}|/T^{actual}$) is $2.121\%$; 2) $92.5\%$ of estimation samples have an error less than $5\%$\footnote{The performance of our model may vary with different DNN models, but our empirical results show limited estimation error in all the architectures we tested.}.

\subsection{Problem Formulation}\label{section:problemformulation}
To fulfill the demand of the real-time capability, we aim to optimize the worst-case latency of all UEs in a min-max fashion.
Therefore the global utitlity function we aim to minimize is defined as:
$$U(S,F) \triangleq \max_{i \in \mathcal{N}} \ T_i(s_i,f_i),$$
where $S$ is the vector of partitioning decisions, i.e. $S\triangleq (s_1,s_2,\cdots,s_n)$, $F$ denotes the vector of computational resources allcation, i.e. $F\triangleq (f_1,f_2,\cdots,f_n)$, the execution latency $T_i(s_i,f_i)$ is given by~\eqref{eq:M1}.

And our problem can be formulated as follows:
\begin{align}
\mathbf{P} : \quad &  \min U(S,F)
\nonumber  \\
 & \sum_{i \in \mathcal{N}} f_i = \beta, \label{eq:c2} \\
 & (1-f_i)k_i\leq s_i \leq k_i \quad \forall i \in \mathcal{N}, \label{eq:c3}\\
 & f_i, \ s_i\ \in \mathbb{N} \quad  \forall i \in \mathcal{N}, \label{eq:c4}
\end{align}
where \eqref{eq:c2} is the constraint on the total number of computational resources on the edge server, \eqref{eq:c4} indicates that $f_i$ and $s_i$ are non-negative integers, \eqref{eq:c3} establishes the constraint on $s_i$ such that (i) $s_i$ is no more than $k_i$, (ii) if the edge server does not assign resource to UE $i$ (i.e., $f_i=0$), $s_i=k_i$, meaning that the entire computation task should be done locally.

It's worth noting that there are approximately $k^n*\frac{(\beta+n-1)!}{(n-1)!\beta!}$ possible solutions in decision space, which is extremely large.
Besides, the generalized problem of \textbf{P} (with arbitrary utility function) is a typical multi-choice knapsack problem (by considering the utility as profit and the assigned resource as weight), which is well known to be NP-hard.
As for our specific utility function, noticing (\ref{eq:M1}) is a nonlinear function, our problem \textbf{P} is by nature a nonlinear integer optimization problem which is notoriously intractable in general.

In the next section, we investigate the structural properties of our utility function and propose an algorithm that can solve \textbf{P} optimally in polynomial time.

\section{Solving P: Iterative Alternating Optimization}\label{section:algorithmdesign}

In this section, we present our proposition, termed as Iterative Alternating Optimization (IAO), to solve \textbf{P}. We first reveal several structural properties of our problem. We then describe our IAO algorithm in detail. Using the established structural properties, we further prove the optimality of our algorithm and analyze its complexity. We complete this section by investigating the impact of estimation error.


\subsection{Structural Properties of \textbf{P}}\label{subsection:properties}

\begin{property}\label{property1}
Given a fixed number of resources $f_i$, we can derive the optimal solution $s_i$ minimizing $T_i(s_i,f_i)$.\end{property}

Since the partitioning decision space is discrete, we can always obtain the individual optimal solution by searching all possible $s_i$ in $O(k_i)$ time, given the fixed assigned computational resource $f_i$.

In the second property we establish below, we consider the case where the resource allocated to UE $i$, $f_i$, is given, and define the corresponding optimal partitioning decision as
$$s_i^*(f_i)\triangleq \arg\min_{s_i} \ T_i(s_i,f_i).$$

\begin{property}\label{property2}
It holds that $T_i(s_i^*,f_i)$ is monotonically non-increasing in $f_i$, i.e.,

$$ f_i^1\ge f_i^2 \longrightarrow T_i(s_i^*(f_i^1),f_i^1) \le T_i(s_i^*(f_i^2),f_i^2) \ \forall i \in \mathcal{N}.$$
\end{property}
\begin{proof}\label{proof:property2}
See Appendix \ref{appendix:property2}.
\end{proof}

Property~\ref{property2} is intuitive to understand by reflecting the common sense that ``more resources do no harm''.

\subsection{The Iterative Alternating Optimization Algorithm}\label{section:algorithm}
The rationale of our iterative alternating optimization algorithm to is that the min-max objective function of \textbf{P} will only benefit from the improvement of the UE with the worst utility. Therefore, we iteratively adjust the allocation of computational resources by reallocating the computational resources to the UE with the worst utility.
Hopefully we can reach an equilibrium where no further improvement can be made, which corresponds to the system optimum.

The pseudo-code of our IAO algorithm is illustrated in Alg.~\ref{alg:1}. We first set the adjustment stepsize denoted by $\tau$, the quantity of computing resource the algorithm adjusts in each iteration. We then choose a random resource allocation vector as the initial strategy and compute $s_i^*$ based on $f_i$. As analyzed in Section \ref{subsection:complexity}, if we have certain prior knowledge, we can select better initial strategy which may reduce the convergence time.

The core part of our algorithm is the main loop, which can be decomposed into three steps. In the first step, we check each UE $j$ if its allocated resource $f_j$ is small than the stepsize and if its latency exceeds the maximum latency among all UEs and mark it as exhausted if yes. If all the UEs are exhausted, indicating we cannot be better off by conducting any resource adjustment, we terminate the algorithm by outputting the current allocation vector. Otherwise we move $\tau$ quantity of resource from a non-exhausted UE with the lowest latency after adjustment to the UE with the highest latency. We then update the related latencies of the touched UEs. The above loop terminates if all the UEs are exhausted, meaning that no UE can improve the maximum latency among UEs by snapping resource from others.


\begin{algorithm}[t]
    \caption{Iterative Alternating Optimization (IAO)}
    \label{alg:1}
    \begin{algorithmic}[1]
        \STATE set adjustment stepsize $\tau\leftrightarrow 1$
        \STATE choose a random resource allocation vector $(S,F)$

        \FOR{each $i\in{\mathcal{N}}$}
            \STATE $s_i^* \leftarrow \arg\min_{s_i} \ T_i(s_i,f_i)$
        \ENDFOR

        \LOOP
            \STATE $L_{max}\leftarrow \max_{i\in {\mathcal{N}}} T_i(s_i^*,f_i)$

            \FOR{each $j\in{\mathcal{N}}$}
                \IF{$f_j-\tau<0$}
                    \STATE mark UE $j$ as exhausted
                    \STATE \textbf{continue}
                \ENDIF
                \STATE $s_j' \leftarrow \arg\min_{s_j} \ T_j(s_j,f_j-\tau)$
                \IF{$T_j(s_j',f_j-\tau)\ge L_{max}$}
                    \STATE mark UE $j$ as exhausted
                \ENDIF
            \ENDFOR

            \IF{all UEs are marked as exhausted}
                \STATE \textbf{return} $(S,F)$
            \ENDIF

            \STATE $i_{min} \leftarrow \min_{i\in{\mathcal{N}}} \ T_i(s_i',f_i-\tau)$
            \STATE $i_{max} \leftarrow \max_{i\in{\mathcal{N}}} \ T_i(s_i^*,f_i)$
            \STATE move $\tau$ quantity of resources from UE $i_{min}$ to $i_{max}$:
            $$f_{i_{max}}\leftarrow f_{i_{max}}+\tau, \ f_{i_{min}}\leftarrow f_{i_{min}}-\tau$$
            \STATE update $s_{i_{max}}^*$ and $s_{i_{min}}^*$:
            $$s_{i_{max}}^* \leftarrow \arg\min_{s_{i_{max}}} \ T_{i_{max}}(s_{i_{max}},f_{i_{max}}),$$
            $$s_{i_{min}}^* \leftarrow \arg\min_{s_{i_{min}}} \ T_{i_{min}}(s_{i_{min}},f_{i_{min}})$$
        \ENDLOOP
    \end{algorithmic}
\end{algorithm}

\subsection{Optimality Analysis}
In this subsection, we prove the optimality of the IAO algorithm.

\begin{theorem}\label{theorem1}
Once terminated, Alg.~\ref{alg:1} outputs an optimal solution of \textbf{P}.
\end{theorem}

\begin{proof}\label{proof:theorem1}
See Appendix \ref{appendix:theorem1}.
\end{proof}

\subsection{Complexity Analysis}\label{subsection:complexity}

In this subsection, we investigate the complexity of our IAO algorithm.

In this regard, we define the \textbf{Manhattan Distance} between any two allocation vectors $(S,F)$ and $(S',F')$ as $D_m \triangleq \sum_{i=1}^n |f_i-f_i'|$.
The Manhattan distance quantifies the difference in resource allocation between $(S,F)$ and $(S',F')$. We prove in the following proposition that the Manhattan distance between any two profiles $(S,F)$ and $(S^*,F^*)$ is upper-bounded by $2\beta$ as long as $\sum_{i \in \mathcal{N}} f_i^* = \sum_{i \in \mathcal{N}} f_i = \beta$, where we recall that $\beta$ is the total amount of computation resource of the edge server.

\begin{ppst}\label{proposition1}
The Manhattan distance between two profiles $(S,F)$ and $(S^*,F^*)$ is upper-bounded by $2\beta$ as long as $\sum_{i \in \mathcal{N}} f_i^* = \sum_{i \in \mathcal{N}} f_i = \beta$.
\end{ppst}
\begin{proof}\label{proof:distance}
It follows from $f_i \geq 0$, $f_i \geq 0, \ \forall i \in {\mathcal{N}}$ and $\sum_{i \in \mathcal{N}} f_i^* = \sum_{i \in \mathcal{N}} f_i = \beta$ that
\begin{align*}
D_m &= \sum_{i \in \mathcal{N}} |f_i-f_i^*| \leq \sum_{i \in \mathcal{N}} |f_i| + \sum_{i \in \mathcal{N}} |f_i^*| \\
&=\sum_{i \in \mathcal{N}} f_i + \sum_{i \in \mathcal{N}} f_i^*=2\beta.
\end{align*}
\end{proof}


Let $(S(t),F(t))$ denote the resource allocation profile derived in iteration $t$ of Alg.~\ref{alg:1}, we prove that the Manhattan distance between $(S(t),F(t))$ and the optimal solution $(S^*,F^*)$ decreases by $2$ every iteration $t$ until when Alg.~\ref{alg:1} is terminated.

\begin{ppst}\label{proposition2}
Denote $D_m(t)$ as the Manhattan distance between $(S(t),F(t))$ and $(S^*,F^*)$. Before the termination of Alg.~\ref{alg:1}, it holds that $D_m(t)-D_m(t+1)=2, \ \forall t$.
\end{ppst}

\begin{proof}\label{proof:reduce2}
See Appendix \ref{appendix:proposition2}.
\end{proof}

Armed with the above two propositions, the following theorem establishes the time complexity of Alg.~\ref{alg:1}.
\rev{
\begin{theorem}\label{theorem2}
Alg. 1 terminates in at most $\beta$ iterations and its time complexity is upper bounded by $O(nk\beta)$.
\end{theorem}
\begin{proof}\label{proof:theorem2}
The first \textbf{for} loop takes $O(nk)$ time.
And each iteration of the main loop (from step 6 to 25) takes $O(nk)$ time, dominated by the second \textbf{for} loop (from step 8 to 17). $D_m$ is the Manhattan distance indicating the difference between initial solution to optimal solution. It follows from Proposition 2 that each iteration shrinks this difference by $2$. Therefore, the total rounds of iteration is $D_m/2$. The time complexity, dominated by the main loop, is $O(nkD_m)$.

Beside, it follows from Proposition 1 that the Manhattan distance between any initial allocation profile $(S(0),F(0))$ is at most $2\beta$ from the optimal one. Hence, Alg.~\ref{alg:1} terminates in at most $\beta$ iterations. Therefore, the total time complexity of Alg.~\ref{alg:1} is upper bounded by $O(nk\beta)$.
\end{proof}

Theorem~\ref{theorem2} demonstrates that the complexity of our algorithm scales linearly in terms of the number of UEs, the quantity of computing resources, and the number of DNN layers.
And it also indicates that a proper initialization (closer to the optimal solution and less Manhattan distance) will lead to less runtime of the algorithm. This feature would help when the algorithm is adapted to an online system since the historical performance can be utilized to perform a better initialization.
}


\subsection{Accelerating Alg.~\ref{alg:1} with Decremental Stepsize $\tau$}

In this subsection, we improve Alg.~\ref{alg:1} by adopting a dynamic stepsize $\tau$ instead of fixing $\tau$ to $1$ using the algorithm named iterative alternating optimization with decremental stepsize (IAO-DS). The rationale behind our idea is to start with a large stepsize to quickly adjust the allocation profile $(S,F)$ towards the optimal and then gradually refine our search by decreasing the stepsize.

\begin{algorithm}[t]
    \caption{Iterative Alternating Optimization with Decremental Stepsize (IAO-DS)}
    \label{alg:imp}
    \begin{algorithmic}[1]
        \STATE set decremental factor $p$, set $q\leftarrow \lfloor \log_p \beta \rfloor$
        \STATE choose random initial allocation profile $(S_0,F_0)$


        \FOR{$i=0$ \textbf{to} $q$}
            \STATE run Alg.~\ref{alg:1} with $(S_{i},F_{i})$ as the initial allocation profile under stepsize $\tau=p^{q-i}$, obtain the output denoted by $(S_{i+1},F_{i+1})$
        \ENDFOR

        \STATE \textbf{return} $(S_{q+1},F_{q+1})$
    \end{algorithmic}
\end{algorithm}

The pseudo-code of IAO-DS algorithm is illustrated in Alg.~\ref{alg:imp}. We start by choosing a decremental factor $p$ which is an integer larger than $1$, e.g., $p=2$. We then enter the \textbf{for} loop. At each iteration, we run Alg.~\ref{alg:1} with the output of the last iteration $(S_{i},F_{i})$ as the initial allocation profile by using stepsize $\tau=p^{q-i}$. By this we gradually approaches the optimal allocation and refine our search by decreasing $\tau$ until the last iteration where we can find the optimal solution with $\tau=1$. Noting that the optimal decremental factor $p$ with fastest convergence speed is hard to derive theoretically, but can be enumerated efficiently in practice since the value of resource amount $\beta$ will not be a large number.

Alg.~\ref{alg:imp} is guaranteed to return an optimal solution because the last iteration of the \textbf{for} loop invokes Alg.~\ref{alg:1} with $\tau=1$, whose optimality is proved in Theorem~\ref{theorem1}. We next establish the complexity of Alg.~\ref{alg:imp}.

\begin{theorem}\label{th:imp}
The time complexity of Alg.~\ref{alg:imp} is bounded by $O(nk\beta)$.
\end{theorem}

\begin{proof}
See Appendix \ref{app:t3}.
\end{proof}

In Theorem \ref{th:imp}, we prove that the theoretical complexity of IAO-DS is still bounded by $O(nk\beta)$, i.e., iteratively adjusting the stepsize would not increase the complexity. Furthermore, extensive experimental results in Section~\ref{subsection:scale} show that IAO-DS can converge much faster than IAO, due to the fact that large stepsizes at the initial iterations can quickly steer the resource allocation towards the optimal solution.

\subsection{Performance Analysis with Estimation Error}


In the theoretical analysis above, we implicitly assume that the estimation of execution latency of a UE is absolutely accurate.
We complete the analysis of our algorithm by investigating a realistic case when the latency model is inaccurate.
As proved in the following, our framework guarantees a bounded performance loss in such case.

We use superscript $E$ to denote the estimated latency given by our model and $A$ for the actual latency. Let $\epsilon$ be the relative estimation error defined as
\begin{eqnarray}
\epsilon \triangleq \frac{|T^{E}-T^{A}|}{T^{A}}.
\label{eq:epsilon}
\end{eqnarray}

\begin{theorem}\label{theorem4}
The relative utility loss of Alg.~\ref{alg:1} and Alg.~\ref{alg:imp} under relative estimation error $\epsilon$ is upper-bounded by $2\epsilon/(1-\epsilon)$.
\end{theorem}
\begin{proof}
See Appendix \ref{app:t4}.
\end{proof}

Theorem~\ref{theorem4} essentially demonstrates that our solution at most doubles the estimation error in the final result under small $\epsilon$. In this regard, it exhibits nice resilience against estimation error.

\section{Performance Evaluation}\label{section:evaluation}

\subsection{Prototype Setup}\label{protosetup}
To evaluate the performance of our proposal, we build a multi-UE edge system prototype.
As shown in Fig. \ref{real}, we use a workstation equipped with a $8$-core $3.7$GHz Intel CPU and $16$G RAM to act as the edge server to provide computing services to UEs.
The UEs are composed of $2$ Raspberry Pis and $2$ NVIDIA Jetson Nanos.
All the UEs are either wirelessly connected to the edge server through Wi-Fi, or directly connected with the edge server by LAN to emulate heterogeneous network conditions for different UEs.

We use the pretrained DNN models from the standard implementation from famous package tensorflow\cite{tensorflow}.
In this prototype, we deploy the lightweight DNN model --  MobilenetV2~\cite{mobilenetv2} on Raspberry Pis, since Raspberry Pi is incapable to load a large-size DNN model like VGG19~\cite{vgg19} due to its 1G memory limit.
Both MobilenetV2 and VGG19 are deployed on  NVIDIA Jetson Nanos.

At the edge server side, we use the Docker containers\cite{docker} as virtual servers to provide the DNN partitioning service to each UE independently.
Multiple CPU cores (i.e., computational resources) are assigned to dockers, and the minimum computational resource unit is set to $0.1$ core.

\subsection{Experiment Settings}\label{sim}
\begin{figure}[t]
    \centerline{\includegraphics[width=7cm]{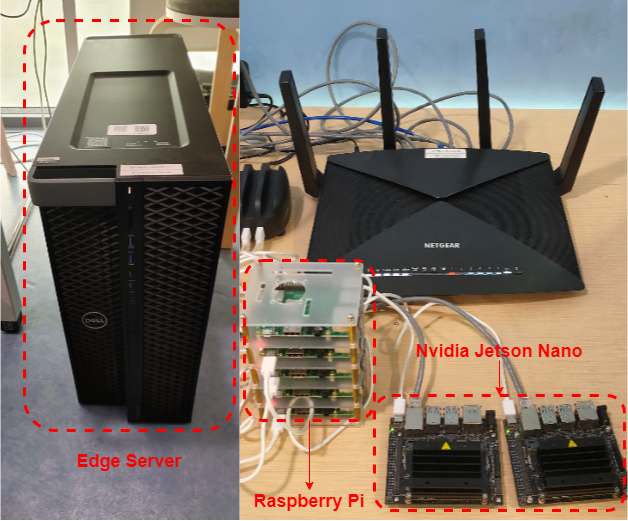}}
    \caption{Experiment Prototype}
    \label{real}
    \end{figure}

\begin{figure*}[t]
    \begin{minipage}[t]{0.5\linewidth}
    \centering
    \includegraphics[width=2.5in]{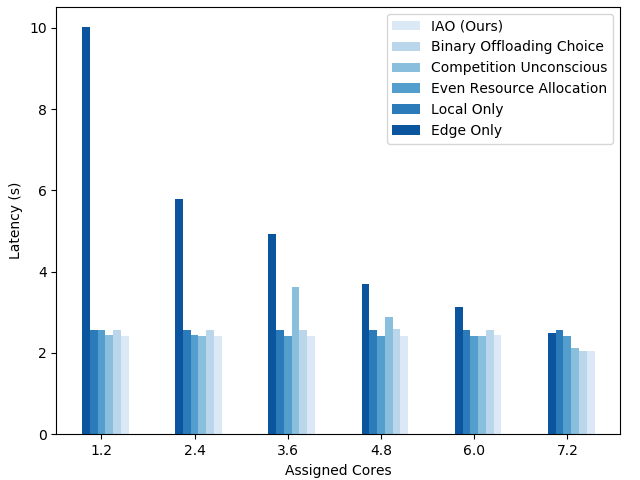}
    \caption{\quad Latency versus the amount of computational\\\quad \quad resources under average bandwidth of 5Mb/s}\label{5kb}
    \end{minipage}%
    \begin{minipage}[t]{0.5\linewidth}
    \centering
    \includegraphics[width=2.5in]{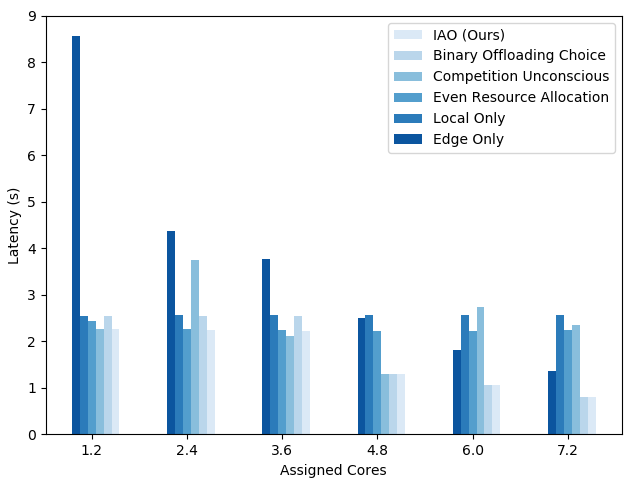}
    \caption{\quad Latency versus the amount of computational\\ resources with high bandwidth (10Mb/s for mobile devices and 100M/s for immobile devices)}\label{100kb}
    \end{minipage}
    \end{figure*}

    \begin{figure*}[t]
    \begin{minipage}[t]{0.5\linewidth}
    \centering
    \includegraphics[width=2.5in]{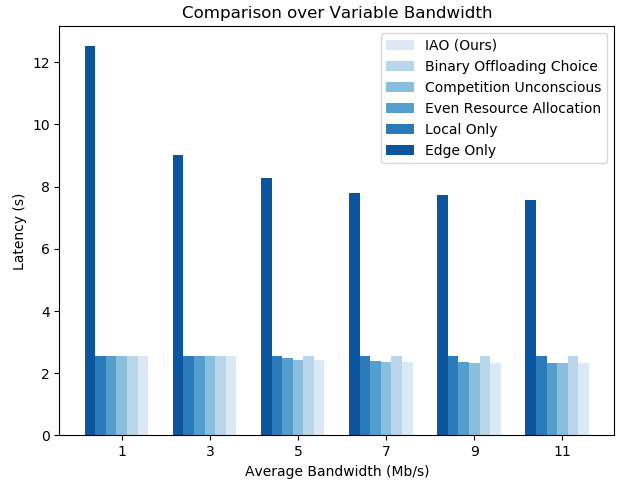}
    \caption{\quad Latency versus average bandwidth\\\ \ \ \quad when edge server has 2 CPU cores}\label{2cores}
    \end{minipage}%
    \begin{minipage}[t]{0.5\linewidth}
    \centering
    \includegraphics[width=2.5in]{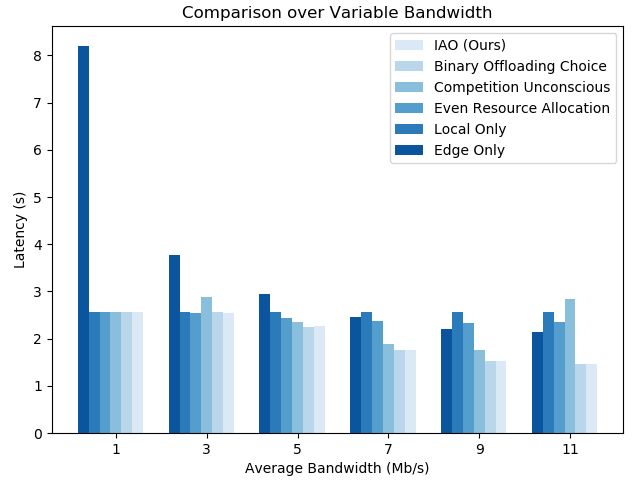}
    \caption{\quad Latency versus average bandwidth\\ when edge server has 7 CPU cores}\label{7cores}
    \end{minipage}
    \end{figure*}
By default, the edge server is cooperating with 4 UEs if not explicitly mentioned.
Two Raspberry Pis running MobilenetV2 are connected to the edge server via WiFi to represent low-end mobile devices, and two NVIDIA Jetson Nanos running VGG19 are connected with the edge server by LAN to act as high-end devices such as smart routers and smart home devices. We will also conduct experimental evaluations with different number of UEs.

\subsection{Benchmarks}\label{benchmarks}
We evaluate our framework by comparing the performance with two state-of-the-art paradigms and three naive approaches as follows.

\textbf{Binary Offloading Choice}. A paradigm from \cite{minmax} that makes a joint decision over task offloading and computational resources allocation. However, their work only considers to offload a task as an entire entity. Thus, a task can only be executed on either the edge server or the local device in a non-cooperative fashion.

\textbf{Even Resource Allocation}. Several works have considered to offload part of the neural network to edge/cloud in a single user case\cite{neurosurgeon}\cite{jdnn}\cite{coin}. We extend them into multi-user case that the edge server fairly allocates the computational resources for each UE evenly.

\textbf{Competition Unconscious}. All users optimize their own DNN partitioning choice according to the amount of computational resource at the edge server without considering the competition of resource allocation against other UEs. After offloading decisions are made, the edge server evenly allocates its computational resources to the offloaded tasks.

\textbf{Local Only}. All users execute their tasks locally.

\textbf{Edge Only}. All users offload their tasks to the edge server, and the edge server is capable to adjust the computational resources assigned to each user.

\subsection{Experimental Results}\label{result}
We first present extensive experiments over different amounts of bandwidth and computational resource on the edge server. Since IAO algorithm achieves the same performance as the IAO-DS algorithm, here in this part we only illustrate the results of the IAO algorithm. We will compare these two algorithms from the scalability point of view in the coming section.

As shown in Figs. \ref{5kb}, \ref{100kb}, \ref{2cores} and \ref{7cores}, generally, the execution latency of UEs benefits from the increase of computational resources and bandwidth.
However, for \textbf{Competition Unconscious} scheme, the increase of resources may sometimes lead to worse performances.
The reason is that, with the increase of resources of the environment, UEs blindly make aggressive offloading decisions to pursue better utilities.
However, when there are multiple UEs competing on the limited resources on the edge server, none of them is able to enjoy the expected performance as there were no competitions and the global utility may become worse than before.

Comparing with the traditional \textbf{Local Execution} and \textbf{Edge Only Execution} schemes, as shown in Fig. \ref{100kb}, our scheme shows a significant improvement of at most \textbf{67.6\%} than \textbf{Local Execution} and \textbf{41\%} than \textbf{Edge Only Execution} when there are enough computational resources and bandwidth.

As for the scheme of \textbf{Even Resource Allocation}, the main drawback is it can not properly assign computational resources to help the neediest UEs with large latencies.
As shown in in Figs. \ref{100kb} and \ref{7cores}, our scheme outperforms \textbf{Even Resource Allocation} for at most \textbf{63\%}.

The scheme of \textbf{Binary Offloading Choice} is basically a weak version of our scheme.
Besides the choices of executing locally and offloading to the edge server, our scheme explores more possible offloading choices with DNN partitioning.
As shown in Figs. \ref{100kb} and \ref{7cores}, when there are abundant resources, the optimal choice could just be sending all tasks to the edge server if needed. However, when the resource becomes scarce as in Figs. \ref{5kb} and \ref{100kb}, our scheme is capable to utilize a very tiny amount of resources to optimize the performance.
From Fig. \ref{100kb} it shows that our scheme at most \textbf{14\%} improvement than \textbf{Binary Offloading Choice} scheme.
This ratio varies according to the architecture of DNN model for whether there are proper positions for DNN partitioning.

\begin{figure}[t]
    \centerline{\includegraphics[width=7.5cm]{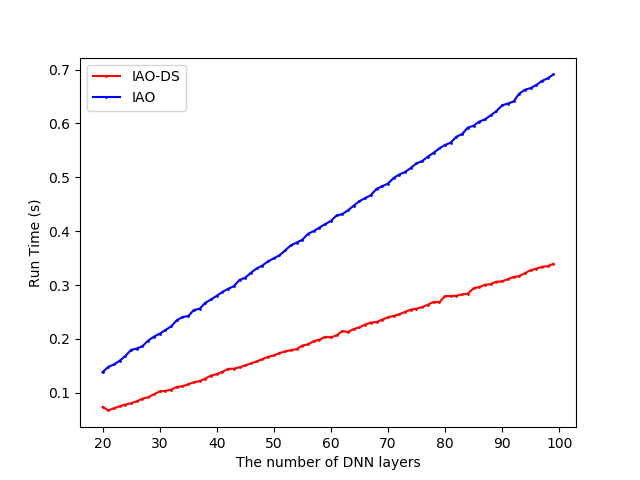}}
    \caption{Run time versus $k$}
    \label{pic:ck}
    \end{figure}

\begin{figure}[t]
    \centerline{\includegraphics[width=7.5cm]{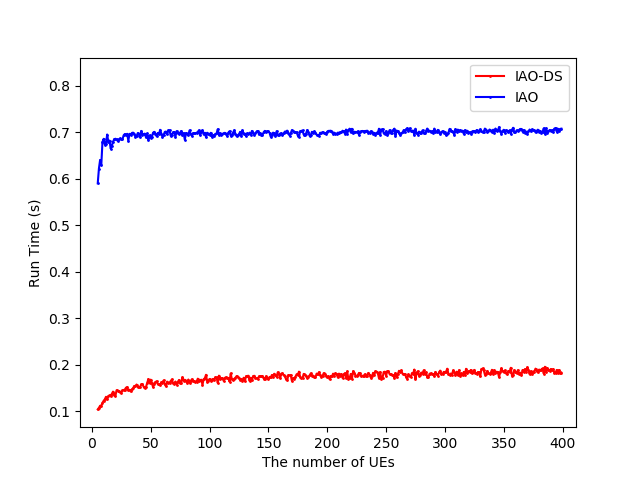}}
    \caption{Run time versus $n$}
    \label{pic:cn}
    \end{figure}

\begin{figure}[t]
    \centerline{\includegraphics[width=7.5cm]{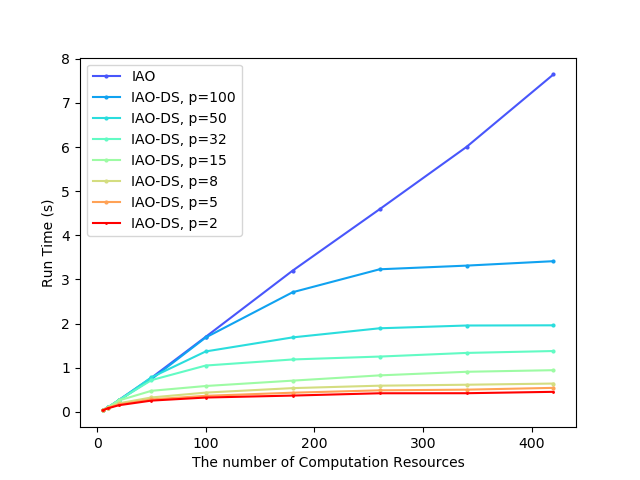}}
    \caption{Run time versus $\beta$}
    \label{twoalgo}
    \end{figure}
\subsection{Scalability}\label{subsection:scale}
\rev{Since IAO and IAO-DS achieve the same execution latency performance, in this section we compare the IAO and IAO-DS algorithms from the perspective of scalability (i.e., the change of convergence time with different factors) in large-scale numerical settings to verify that our algorithms are competent for large-scale practical deployments.} We conduct detailed evaluations as follows.

\textbf{Scale increase by the number of layers}.
$k$ is the number of logical layers of our neural network (i.e., the number of possible offloading decisions in decision space).
As we analyzed in Section \ref{subsection:complexity}, the execution time of IAO should increase linearly with $k$ since we always traverse over $k$ offloading decisions in any iteration. Fig. \ref{pic:ck} shows the practical increase with $k$ which obeys our theoretical analysis. Also, IAO-DS runs much faster than IAO, with more than 50\% convergence time reduction.

\textbf{Scale increase by the number of UEs}.
$n$ is the number of UEs. As shown in Fig. \ref{pic:cn}, the run time increases rapidly at first and however, its gradient decrease immediately after some certain boundary. This happens because when $n$ grows larger as $\beta$ stays unchanged, the average resource for each UE is getting smaller, and most UEs with no resources do not participate in resource reallocation step as mentioned in step 8-12 in Algorithm 1. That leads to a low increasing rate of run time. Also, we observe that the convergence time of IAO-DS is much smaller than that of IAO.

\textbf{Scale increase by the amount of computational resource}.
$\beta$ is the number of MCRU on the edge server, intuitively, if there are more resources to allocate, it takes more time to make a decision. In Fig. \ref{twoalgo}, we also observe that IAO-DS shows a significant reduction in run time than IAO. For IAO-DS, it achieves the fastest run time when the decremental factor $p=2$, due to the rich generated stepsizes for sufficient resource allocation profile adjustment.

\rev{Note that, similar to many existing studies, to enable tractable
analysis we consider that the edge server conducts the DNN tasks scheduling in the batch-by-batch manner (batch size can be set dynamically according to the time-varying task arriving rates). This
can be achieved via a task buffering mechanism (which is widely
used in cloud computing) in practice.

Also, our algorithm has a linear convergence time in terms of the number of device users, and hence our algorithm can work well if the set of users in the system change at a slower time-scale (e.g., seconds/minutes) than the convergence time by the algorithm (e.g., tens of milliseconds).

The algorithm design with efficient worst-case performance guarantee for the online scheduling scenarios with very fast-change changing system dynamics at the time-scale of milliseconds is much more challenging, and we will consider such challenging case in a future work.}

\section{Conclusion}\label{section:conclusion}
In this paper, we have proposed a framework to jointly address multi-user DNN partitioning and computational resource allocation. The central component of our framework is an algorithm minimizing the maximum delay among all UEs. Our algorithm is shown to achieve system optimum with polynomial time complexity, and exhibits nice resilience against estimation error.
We also validate our proposition by building a prototype in the multi-user MEC environment for the task of device-edge cooperative DNN partitioning.
By extensive experiments on real environment, we show that our framework outperforms the existing offloading schemes and scales effectively.

\rev{For the future work, we are going to consider the more challenging scenarios of the joint computing, communication and energy resource allocation for multi-user edge intelligence applications.}

\appendices
\section{Proof of Property \ref{property2}}\label{appendix:property2}
\begin{proof}\label{proof:property2}
    It follows from~\eqref{eq:M1} that
    $$T_i(s_i^*(f_i^2),f_1) \leq T_i(s_i^*(f_i^2),f_2).$$
    By the definition of $s_i^*(f_i^1)$, we have
    $$T_i(s_i^*(f_i^1),f_1) \leq T_i(s_i^*(f_i^2),f_1).$$
    Therefore we have
    $$T_i(s_i^*(f_i^1),f_1) \leq T_i(s_i^*(f_i^2),f_1) \leq T_i(s_i^*(f_i^2),f_2).$$
    Property~\ref{property2} is proved.
    \end{proof}

\section{Proof of Optimality of IAO}\label{appendix:theorem1}
\begin{proof}\label{proof:theorem1}
Assume, by contradiction, that there exists a better solution $(S',F')\triangleq (s_1',s_2',\cdots, s_n',f_1',f_2',\cdots, ,f_n')$ than the output of Alg.~\ref{alg:1} denoted by $(S,F)$, i.e.,
\begin{equation}
U(S',F')<U(S,F).\label{eq:a1}
\end{equation}
We denote $p$ as the UE with the largest latency $T_p(s_p,f_p)$ among all $T_i(s_i,f_i)$  ($\forall i \in N$). By the definition of $U(S,F)$, we have
\begin{equation}
U(S,F) = T_p(s_p,f_p).\label{eq:d1.1}
\end{equation}
By definition, we have $U(S',F')=\max_{i \in \mathcal{N}} T_i(s_i',f_i')$. It then follows from~\eqref{eq:a1} and~\eqref{eq:d1.1} that
\begin{equation}
T_p(s_p',f_p')\le U(S',F')<U(S,F)=T_p(s_p,f_p).\label{eq:a1.1}
\end{equation}
If $f_p' = f_p$, following from Alg.~\ref{alg:1} that $s_p$ is optimal w.r.t. $f_p$, it holds that $T_p(s_p',f_p')\ge T_p(s_p,f_p)$, which contradicts with~\eqref{eq:a1.1}. Therefore we have $f_p' \ne f_p$. It then follows from \eqref{eq:a1.1} and Property~\ref{property2} that
\begin{equation}
f_p' > f_p.\label{eq:a1.2}
\end{equation}
Noticing $\sum_{i \in \mathcal{N}} f_i' = \sum_{i \in \mathcal{N}} f_i = \beta$, there must exist another UE $q$ such that
\begin{equation}
    f_q' < f_q.\label{eq:a1.3}
    \end{equation}
Recall that Alg.~\ref{alg:1} terminates when all UEs are ``exhausted'', i.e.,
$$T_j(s_j^{*},f_j-1) \geq L_{max} \quad \quad \forall j \in N,$$
where $L_{max}=U(S,F)$, the largest individual latency $T(s,f)$. Recall the definition of $p$, we have
\begin{equation}
    L_{max} = T_p(s_p,f_p)=U(S,F),\label{eq:d1.2}
    \end{equation}
and
\begin{equation}
    T_j(s_j^{*},f_j-1) \geq L_{max} \quad \quad \forall j \in N.\label{eq:d1.30}
    \end{equation}
Particularly from~\eqref{eq:d1.30}, we have
\begin{equation}
    T_q(s_q,f_q-1) \geq L_{max}.\label{eq:d1.3}
    \end{equation}
Recall \eqref{eq:a1.3} that $f_q' < f_q$. For both $f_q$ and $f_q'$ being integers, we have $f_q' \leq f_q-1$.
It then follows from Property~\ref{property2} and~\eqref{eq:d1.3} that
\begin{equation}
    T_q(s_q',f_q')\geq T_q(s_q,f_q-1) \geq L_{max}.\label{eq:a1.4}
    \end{equation}
Combining~\eqref{eq:a1.4} and~\eqref{eq:d1.2} yields
$$U(S',F') \geq T_q(s_q',f_q') \geq L_{max} = T_p(s_p,f_p) = U(S,F),$$
which contradicts with the assumption~\eqref{eq:a1}. This contradiction proves the optimality of Alg.~\ref{alg:1}.
\end{proof}

\section{Proof of Proposition \ref{proposition2}}\label{appendix:proposition2}
\begin{proof}\label{proof:reduce2}
    To make the notation concise, denote
    $$(S,F)\triangleq (S(t),F(t)), \ (S',F')\triangleq (S(t+1),F(t+1)).$$
    Denote $i\triangleq \arg\max_{q\in{\mathcal{N}}} T_q(s_q,f_q)$ and $j\triangleq \arg\min_{q\in{\mathcal{N}}} T_q(s_q,f_q-1)$. We have
    \begin{equation}
        T_i(s_i,f_i)\geq T_q(s_q,f_q) \quad \forall q \in \mathcal{N},\label{eq:d2.1}
        \end{equation}
    \begin{equation}
        T_j(s_j,f_j-1)< T_q(s_q,f_q-1) \quad \forall q \in \mathcal{N}, q\ne j.\label{eq:d2.2}
        \end{equation}
    It follows from Line 23 of Alg.~\ref{alg:1} that
    \begin{equation}
        f_i'=f_i+1, f_j'=f_j-1.\label{eq:d2.25}
        \end{equation}
    Recall that the UE with the largest latency is marked exhausted, which indicates that it can not lose resource in any later iteration, we have
    \begin{equation}
        f_i<f_i'\leq f_i^*.\label{eq:d2.3}
        \end{equation}
    We then have
    \begin{equation}
    \begin{aligned}
    \Delta D_m &\triangleq D_m(t+1) - D_m(t) \\
    &= \sum_{q\in{\mathcal{N}}} |f_q'-f_q^*| - \sum_{q\in{\mathcal{N}}} |f_q-f_q^*| \\
    &= |f_i'-f_i^*|+|f_j'-f_j^*|-|f_i-f_i^*|-|f_j-f_j^*|\\
    &= |f_j'-f_j^*|-|f_j-f_j^*| -1.
    \end{aligned}\label{eq:d2.4}
    \end{equation}
    We next prove $f_j> f_j^*$. To that end, assume by contradiction that
    \begin{equation}
        f_j\leq f_j^*.\label{eq:a2}
        \end{equation}
    It follows from Property~\ref{property2} that
    \begin{equation}
        T_j(s_j^*,f_j^*-1)\leq T_j(s_j,f_j-1).\label{eq:a2.1}
        \end{equation}
    Recall $\sum_{i \in \mathcal{N}} f_i^* = \sum_{i \in \mathcal{N}} f_i = \beta$, there must exist a UE $x$ such that
    $$f_x>f_x^*,$$
    which further leads to $f_x-1 \geq f_x^*$ as they are integers.
    It then follows from Property 2 that
    \begin{equation}
        T_x(s_x,f_x-1)\leq T_x(s_x^*,f_x^*).\label{eq:a2.2}
        \end{equation}
    Recall that at an optimal solution all UEs are exhausted (otherwise we can perform adjustment to decrease the maximum delay among UEs), i.e.,
    \begin{equation}
        T_q(s_q^*,f_q^*-1) \geq  \max_{p \in \mathcal{N}} T_p(s_p^*,f_p^*) \quad \forall q \in \mathcal{N}.\label{eq:d2.4}
        \end{equation}
    It then follows from~\eqref{eq:a2.2} and~\eqref{eq:d2.4} that
    $$T_j(s_j^*,f_j^*-1) \geq \max_{p \in \mathcal{N}}T_p(s_p^*,f_p^*) \geq T_x(s_x^*,f_x^*) \geq T_x(s_x,f_x-1),$$
    which contradicts with~\eqref{eq:d2.2}. We thus have: $f_j > f_j^*$.
    Combining with~\eqref{eq:d2.25}, we have
    \begin{equation*}
    \begin{aligned}
    \Delta D_m &= |f_j'-f_j^*|-|f_j-f_j^*| -1 = -2.
    \end{aligned}
    \end{equation*}
    The proposition is thus proved.
    \end{proof}

\section{Proof of Theorem \ref{th:imp}}\label{app:t3}

    \begin{proof}
        By regarding $\tau$ units of resource as a bundle, there are at most $\left\lfloor \frac{\beta}{\tau} \right\rfloor$ bundles that can be adjusted w.r.t. stepsize $\tau$. By adapting the analysis in Sec.~\ref{subsection:complexity}, we can easily show that running Alg.~\ref{alg:1} with stepize $\tau$ needs at most $O(nk\left\lfloor \frac{\beta}{\tau} \right\rfloor)$. The time complexity of Alg.~\ref{alg:imp} can be derived by summing $O(nk\left\lfloor \frac{\beta}{\tau} \right\rfloor)$ from $\tau=p^q$ to $1$, i.e.:
        \begin{equation*}\nonumber
        \begin{aligned}
        \sum_{i=0}^q nk\left\lfloor \frac{\beta}{p^i} \right\rfloor&= nk(\lfloor \frac{\beta}{p^q} \rfloor + \lfloor \frac{\beta}{p^{q-1}} \rfloor +...+\lfloor \frac{\beta}{1} \rfloor)\\
        &\leq nk(\lfloor \frac{p^{q+1}}{p^q} \rfloor + \lfloor \frac{p^{q+1}}{p^{q-1}} \rfloor +...+\lfloor \frac{p^{q+1}}{1} \rfloor)\\
        &\leq nk(p + p^{2} +...+p^{q+1})\\
        &= nk\frac{p(p^q-1)}{p-1} \leq nk(\beta-1)\frac{p}{p-1}.\\
        \end{aligned}
        \end{equation*}
        Noting $p$ is a chosen constant, therefore the total time complexity of IAO-DS is bounded by
        $O(nk\beta)$.
        \end{proof}

\section{Proof of Theorem \ref{theorem4}}\label{app:t4}

    \begin{proof}
        We prove the theorem for Alg.~\ref{alg:1}. Alg.~\ref{alg:imp} can be handled in the same way.

        Denote the estimated latency and the actual latency of UE $i$ using the solution derived from Alg.~\ref{alg:1} as $T_i^E$ and $T_i^A$, respectively.
        Denote the actual latency and the estimated latency of UE $i$ of the optimal solution as $T_i^{A*}$ and $T_i^{E*}$, respectively. Further define the following utilities:
        \begin{eqnarray}
        U^A=\max_{i \in \mathcal{N}} \ T_i^A, \label{eq:2}
        \\
        U^{E}=\max_{i \in \mathcal{N}} \ T_i^{E}, \label{eq:3}
        \\
        U^{A*}=\max_{i \in \mathcal{N}} \ T_i^{A*}, \label{eq:4}
        \\
        U^{E*}=\max_{i \in \mathcal{N}} \ T_i^{E*}.
        \label{eq:5}
        \end{eqnarray}
        We have shown that
        \begin{eqnarray}
        U^{E} \leq U^{E*}.
        \label{eq:6}
        \end{eqnarray}

        Denote the UEs $i$, $j$, $k$, $l$ as the UEs with the largest latency to the utility of eqs.~\eqref{eq:2}-~\eqref{eq:5}, respectively, i.e.,:
        \begin{eqnarray}
        U^A=T_i^A, \ U^E=T_j^E, \ U^{A*}=T_k^{A*}, \ U^{E*}=T_l^{E*}.
        \label{eq:7}
        \end{eqnarray}
        It follows from~ (\ref{eq:epsilon}) that
        \begin{eqnarray}
        (1-\epsilon)T^{A} \leq T^{E} \leq (1+\epsilon)T^{A},
        \label{eq:8}
        \frac{T^{E}}{1+\epsilon} \leq T^{A} \leq \frac{T^{E}}{1-\epsilon}.
        \label{eq:9}
        \end{eqnarray}
        Combining eqs.~\eqref{eq:6}-~\eqref{eq:9} yields
        \begin{align}
        U^A&=T_i^A \leq \frac{T_i^{E}}{1-\epsilon} \leq \frac{U^E}{1-\epsilon} \leq \frac{U^{E*}}{1-\epsilon}=\frac{T_l^{E*}}{1-\epsilon} \nonumber \\
        &\leq \frac{(1+\epsilon)T_l^{A*}}{1-\epsilon} \leq \frac{(1+\epsilon)U^{A*}}{1-\epsilon}.
        \label{eq:10}
        \end{align}

        Armed with the above results, we can then upper-bound the relative utility loss of Alg.~\ref{alg:1} under the relative estimation error $\epsilon$ as
        $$\frac{|U^A-U^{A*}|}{U^{A*}} \leq \frac{2\epsilon}{1-\epsilon}.$$
        The theorem is proved.
        \end{proof}
%


\bibliographystyle{IEEEtran}
\bibliography{ref}

\end{document}